\documentclass{llncs}

\usepackage[title]{appendix}
\let\cleardoublepage\clearpage

\usepackage{amsfonts}
\usepackage{amsmath}
\usepackage{float}
\usepackage{graphicx}

\usepackage{subcaption}
\usepackage{balance}

\floatstyle{ruled}
\newfloat{algorithm}{th}{lop}
\floatname{algorithm}{Algorithm Sketch}

\spnewtheorem*{proofsketch}{Proof Sketch}{\itshape}{}

\begin{document}

\title{Network Scaffolding for Efficient Stabilization of the \textsc{Chord} Overlay Network\thanks{An early version of this work appeared as a Brief Announcement in SPAA 2021}}

\author{Andrew Berns}
\institute{Department of Computer Science\\University of Northern Iowa\\Cedar Falls, IA USA\\
\email{andrew.berns@uni.edu}}

\maketitle


\begin{abstract}
  Overlay networks, where nodes communicate with neighbors over logical links consisting of zero or more physical links, have become an important part of modern networking.  From data centers to IoT devices, overlay networks are used to organize a diverse set of processes for efficient operations like searching and routing.  Many of these overlay networks operate in fragile environments where faults that perturb the logical network topology are commonplace.  Self-stabilizing overlay networks offer one approach for managing these faults, promising to build or restore a particular topology from any weakly-connected initial configuration.
  
  Designing efficient self-stabilizing algorithms for many topologies, however, is not an easy task.  For non-trivial topologies that have desirable properties like low diameter and robust routing in the face of node or link failures, self-stabilizing algorithms to date have had at least linear running time or space requirements.  In this work, we address this issue by presenting an algorithm for building a \textsc{Chord} network that has polylogarithmic time and space complexity.  Furthermore, we discuss how the technique we use for building this \textsc{Chord} network can be generalized into a ``design pattern'' for other desirable overlay network topologies.
  
  
\end{abstract}

\section{Introduction}
As computers and network connectivity have become an ubiquitous part of society, the size and scope of distributed systems has grown.  It is now commonplace for these systems to contain hundreds or even thousands of computers spread across the globe connected through the Internet.  To better facilitate common operations for applications, like routing and searching, many distributed systems are built using \emph{overlay networks}, where connections occur over logical links that consist of zero or more physical links.  Overlay networks allow nodes to embed a predictable topology onto their (usually fixed) physical topology, selecting the best network for the application's particular needs.

Complicating the use of these overlay networks, however, is the reality that systems composed of such a wide variety and distribution of devices are more prone to failures caused from problems with the devices or physical links.  For instance, fiber optic cables can be severed, power outages can cause machines to disconnect without warning, and even intentional user actions like joining or leaving the system on a predictable schedule can result in an incorrectly-configured overlay network causing the client application to fail.

One approach for managing these faults and preventing failures is to design protocols which are resilient to a targeted set of specific system faults, such as nodes joining or leaving the system.  However, the unpredictable nature of these distributed systems makes it difficult to identify and control for every possible fault.  It is for this reason that researchers have turned to \emph{self-stabilizing overlay networks}.  A self-stabilizing overlay network guarantees that after \emph{any} transient fault, a correct topology will eventually be restored.  This type of network can ensure autonomous operation of distributed systems even in the face of a variety of unforeseen transient faults.

\subsubsection{The Problem} Our focus is on building robust self-stabilizing overlay networks efficiently.  More specifically, we are interested in creating \emph{efficient} algorithms that add and delete logical edges in the network to transform an arbitrary weakly-connected initial topology into a correct \emph{robust} topology.  By \emph{efficient}, we mean these algorithms have a time \emph{and} space complexity which is polylogarithmic in the number of nodes in the network.  By \emph{robust}, we mean topologies where the failure of a few nodes is insufficient to disconnect the network.


\subsubsection{Main Results and Significance}
With this paper, we present an efficient self-stabilizing overlay network with desirable practical properties like robustness and low diameter.  In particular, we present a self-stabilizing algorithm for the creation of a \textsc{Chord} network which has expected polylogarithmic space and time requirements.  Note that this is the first work to present an efficient (in terms of time and space) self-stabilizing overlay network for a \emph{robust} topology.  Note that, while our algorithm is deterministic, it depends upon the prior work of Berns~\cite{berns_avatar-2016}, which was randomized, and therefore our results are in expectation.

Our second result is the explicit identification of a ``design pattern'' we call \emph{network scaffolding} for creating self-stabilizing overlay networks.  This pattern has been used in several other works, and the success of this approach, both previously as well as in this work for building \textsc{Chord}, leads us to believe it can be used for many other topologies as well.  Our work is a first step towards fully defining and analyzing this design pattern.  Our goal is that explicit identification of this design pattern can be useful to other researchers and practitioners in the design and implementation of other self-stabilizing overlay networks.

A preliminary version of this work appeared as a brief announcement at SPAA 2021~\cite{berns_chord_spaa_2021}.  In this version, we provide a more detailed discussion of the \textsc{Avatar} background, an improved analysis showing a better bound on the degree expansion, and also provide an extended discussion about the identified design pattern.

\subsubsection{Related Work and Comparisons}
\label{section:related_work}
The past few decades have seen tremendous growth in both the theory and practice of overlay networks.  Some of this work has focused on \emph{unstructured} overlay networks where connections need not satisfy any particular property and there are no constraints on what is considered a ``legal'' topology, such as Napster and Gnutella~\cite{saroiu_napster_2003}.  There are also \emph{semi-structured} overlay networks, where the network topology must satisfy a particular constraint or has a particular property such that there are several possible correct topologies.  For instance, an overlay may be built to mimic a ``small world'' graph~\cite{hui_swop_2004}.

Our work focuses on \emph{structured} overlay networks, where there is exactly one correct topology for any given set of nodes.  While constructing and maintaining the correct topology adds additional work for the algorithm designer, common operations such as routing and searching are much more efficient with these structured networks.  There are many examples of such overlays, including \textsc{Chord}~\cite{stoica_chord_01}, \textsc{Pastry}~\cite{rowstron_pastry_2001}, and \textsc{Tapestry}~\cite{zhao_tapestry_2002}.  These early examples of structured networks, however, provided very limited fault tolerance.

To this end, much previous work has focused on improving the fault tolerance of overlay networks.  One approach has looked at \emph{self-healing} networks, where the network can maintain certain properties while a limited number of faults occur during a fixed time period.  Examples of this work include the Forgiving Tree~\cite{hayes_forgiving-tree_2008}, the Forgiving Graph~\cite{hayes_forgiving-graph_2009}, and DEX~\cite{pandurangan-dex-2016}.  Many of these approaches also use virtual nodes~\cite{trehan_virtual-2012} as done in our work.  More recently, Gilbert et al. presented \textsc{DConstructor}~\cite{gilbert_dconstructor_2020} which is able to build a correct topology from any initial topology and maintain this in the face of some joins and leaves.  G\"{o}tte et al.~\cite{gotte_time-optimal_podc21} also presented an algorithm for transforming a constant-degree network into a tree in $\mathcal{O}(\log n)$ rounds.  The key difference between these works and ours, however, is we use a stronger fault model, requiring our algorithm to build the correct configuration regardless of the initial topology \emph{or} the initial state of the nodes.  This is a paradigm called self-stabilization, which we discuss next.

\emph{Self-stabilizing overlay networks} are those that guarantee a legal configuration will be automatically restored by program actions after \emph{any} transient fault so long as the network is not disconnected.  This is often modeled as the ability for the network to form a correct topology when starting from an arbitrary weakly-connected state.  One of the first such examples of a self-stabilizing overlay network was for the simple structured \textsc{Linear} topology~\cite{onus_linear_07} where nodes were arranged in a ``sorted list''.  Since then, there have been several other self-stabilizing structured overlay networks created, including \textsc{Skip+}~\cite{jacob_skipplus_09} and \textsc{Re-Chord}~\cite{kniesburges_rechord_11}.  Unlike the simple \textsc{Linear} topology, \textsc{Skip+} and \textsc{Re-Chord} maintain several desirable properties for client applications, including low node degree and low diameter.  Unfortunately, their worst-case time (in the case of \textsc{Re-Chord}) or space (in the case of \textsc{Skip+}) complexity is linear in the number of nodes.

To date, we are aware of only one self-stabilizing overlay network that is efficient in terms of both time and space.  Berns presented the \textsc{Avatar} network framework as a mechanism for ensuring a faulty configuration is detectable for a wide variety of networks, and also gave a self-stabilizing algorithm for the construction of a binary search tree~\cite{berns_avatar-2016}.  Our current work builds upon this \textsc{Avatar} network.

A goal of our current work is to identify a general ``design pattern'' which can be used for building self-stabilizing overlay networks.  There has been little work done on identifying these general patterns for overlay network construction.  One exception to this is the \emph{Transitive Closure Framework}~\cite{berns_tcf_11} (TCF), which provides a way to build any locally-checkable topology by detecting a fault, forming a clique, and then deleting those edges which are not required in the correct configuration.  While TCF can create any locally-checkable topology quickly, it requires node degrees to grow to $\mathcal{O}(n)$ during convergence and is therefore not practical for large networks.

While we are interested in structured topologies, there has been work in semi-structured networks.  For instance, self-stabilizing semi-structured networks such as a \emph{small world network}~\cite{kniesburges_small-world_2012} or a \emph{power law network}~\cite{alsulaiman_power-law_2015} have been proposed.  G{\"o}tte et al~\cite{gotte_sirocco19} recently presented an algorithm for quickly converting certain topologies into ones with a logarithmic diameter.  

%
%
%
%
%
%
%
%
%
%
%
%
%
%
%
%
%
%
%
%
%
%
%
%
%
%
%
%
%
%
%
%
%
%
%
%
%
%
%
%
%
%
%
%
%

\section{Preliminaries}
\label{section:prelims}
\subsection{Model of Computation}
We model our distributed system as an undirected graph $G = (V, E)$, with processes being the $n$ nodes of $V$ and the communication links being the edges $E$.  Each node $u$ has a unique identifier $u.id \in \mathbb{N}$, which is stored as immutable data in $u$.  Where clear from the context, we will use $u$ to represent the identifier of $u$.

Each node $u \in V$ has a \emph{local state} consisting of a set of variables and their values, along with its immutable identifier $u.id$.  A node may execute \emph{actions} from its \emph{program} to modify the values of the variables in its local state.  All nodes execute the same program.  Besides modifying its local state, a node can also communicate with its neighbors.  We use the \emph{synchronous message passing} model of computation with bounded communication channels, where computation proceeds in synchronous rounds.  During each round, a node may receive messages sent to it in the previous round from any of its neighbors, execute program actions to update its state, and send messages to any node in its neighborhood $N(u) = \{v \in V:(u,v) \in E\}$.  We assume the communication channels are reliable with bounded delay so that a message is received in some round $i$ if and only if it was sent in round $i - 1$.

In the overlay network model, nodes communicate over logical links that are part of a node's state, meaning a node may execute actions to create or delete edges in $G$.  In particular, in any round a node may delete any edge incident upon it, as well as create any edge to a node $v$ which has been ``introduced'' to it from some neighbor $w$, such that $(u,w)$ and $(w, v)$ are both in $E$.  Said in another way, in a particular round a node may connect its neighbors to one another by direct logical links.

The goal for our computation is for nodes to execute actions to update their state (including modifying the topology by adding and deleting edges) until a legal configuration is reached.  A \emph{legal configuration} can be represented as a predicate over the state of the nodes in the system, and as links are part of a node's state for overlay networks, a legal configuration is defined at least in part by the network topology.  The \emph{self-stabilizing overlay network problem} is to design an algorithm $\mathcal{A}$ such that when executing $\mathcal{A}$ on a connected network with nodes in an arbitrary state, eventually a legal configuration is reached.  This means that a self-stabilizing overlay network will always automatically restore a legal configuration after any transient fault so long as the network remains connected.

\subsection{Performance Metrics}
We analyze the performance of our self-stabilizing overlay network algorithms in terms of both time and space.  For time, we are interested in how quickly the network will be able to recover from a transient fault.  Specifically, we measure the maximum number of (synchronous) rounds that may be required in the worst case to take any set of $n$ nodes from an arbitrary connected configuration to a legal configuration.  This is called the \emph{convergence time}.

The space complexity measure of interest for us is related to the maximum number of neighbors a node might have during convergence that were not present in the initial configuration and are not required in the final configuration.  Said in another way, we are interested in the number of ``extra'' neighbors a node may acquire during convergence.  More specifically, we use the \emph{degree expansion} metric~\cite{berns_avatar-2016}, which is the ratio of the maximum node degree of any node during convergence over the maximum node degree from the initial or final configuration.  The degree expansion helps us to judge how much a node's degree increases as a result of program actions and not from the requirements of the final topology or from a high-degree initial configuration.

Finally we note that, as with many distributed algorithms, we consider an efficient algorithm one which keeps these measures polylogarithmic in the number of nodes in the network.

\section{\textsc{Avatar}}
\label{section:avatar}
Our efficient algorithm for creating the \textsc{Chord} network builds upon previous work on the \textsc{Avatar} overlay network framework.  While a full discussion of this network can be found in the work by Berns~
\cite{berns_avatar-2016}, we present a brief summary and discussion below to provide the necessary background to understand the new contributions of our own work.

\subsection{The \textsc{Avatar} Overlay Network}
The \textsc{Avatar} overlay network framework can be used to define a variation of any particular network topology.  The general idea behind the \textsc{Avatar} framework is to create a dilation-1 embedding of a particular $N$ node \emph{guest network} (with node identifiers in the range of $[0, N)$) onto the $n$ node \emph{host network} (with $n \leq N$).  More specifically, in \textsc{Avatar} each node $u \in V$ from the host network (except the two nodes with the smallest and largest identifiers) simulates or ``hosts'' all nodes from the guest network with identifiers in $u$'s \emph{responsible range}, defined as the range of $[u.id, v.id)$, where $v.id$ is the smallest identifier greater than $u.id$ taken from all nodes in $V$, which we call the \emph{successor} of $u$.  The node $u_0$ with the smallest identifier has a responsible range of $[0, v.id)$ (where again $v$ is the successor of $u_0$), while the node with the largest identifier $u_{n-1}$ has a responsible range of $[u_{n-1}.id, N)$.  To ensure a dilation-1 embedding, for every edge $(a, b)$ in the guest network there exists an edge between the host nodes of $a$ and $b$ in the host network, or the host node for $a$ and $b$ is the same -- that is, either both $a.id$ and $b.id$ are in the responsible range of the same host node, or there exists the edge $(host(a), host(b))$ in the host network such that $a.id$ is in the responsible range of host node $host(a)$ and $b.id$ is in the responsible range of host node $host(b)$.  The definition of some $N$ node guest network $\textsc{Guest}(N)$ along with the constraints on the corresponding edges in the host network define the legal $\textsc{Avatar}(\textsc{Guest}(N))$ network.

%

The use of a guest and host network provides two advantages.  First, the requirement that there is exactly one correct configuration for any given $N$ (meaning the guest network uses nodes $[0, N)$), along with the fact that the successor relationship used in the host network to define the responsible ranges can easily be determined from a node's local state, ensures that any topology is \emph{locally checkable}.  Second, we can design our algorithms (for both stabilization and end-user applications) to execute on the guest network, which has a single predictable configuration for a given $N$, regardless of the node set $V$.  This simplifies both the design and analysis of our algorithms.

As a final comment, we note that the use of \textsc{Avatar} does require all nodes to know $N$, the upper bound on the number of nodes in the network, and that our analysis of convergence time and degree expansion is in terms of this $N$.  Given that all of our algorithms have polylogarithmic time and space requirements, in practice one could easily select an $N$ which was large enough to accommodate any possible node additions while still having a time and space complexity less than many existing algorithms which have complexity at least linear in the actual number of nodes $n$.  Said in another way, even when $N$ is much larger than $n$, our efficient algorithms may still require fewer resources if $\log N \ll n$.  If we consider IPv6, for instance, $\log N$ would be only $128$.

\subsection{\textsc{Avatar}(\textsc{CBT})}
Beyond defining the \textsc{Avatar} framework, Berns also defined the \textsc{Cbt} guest network and a self-stabilizing algorithm for building the $\textsc{Avatar}(\textsc{Cbt})$ network in a polylogarithmic number of rounds with a polylogarithmic degree expansion (both in expectation).  The \textsc{Cbt} topology is simply a complete binary search tree of the specified $N$ nodes.  

\subsubsection{Stabilization}
Our work building the \textsc{Chord} network depends upon the existence of a ``scaffold'' \textsc{Cbt} network.  To be self-stabilizing, we need a way to build this \textsc{Cbt} in a self-stabilizing manner.  This is exactly what is provided by the earlier \textsc{Avatar} work of Berns.  The full description of the self-stabilizing algorithm for $\textsc{Avatar}(\textsc{Cbt})$ can be found in the original work~\cite{berns_avatar-2016}.  We present a short informal summary of the algorithm's operation here to assist in understanding and verifying the correctness of our approach.

The general idea for the $\textsc{Avatar}(\textsc{Cbt})$ algorithm can be described using three components:
\begin{enumerate}
    \item \emph{Clustering}: The first step in the algorithm is for nodes to form clusters.  These clusters begin as a single host node hosting a full $N$ virtual node \textsc{Cbt} network.  In the initial configuration, nodes may not be a part of a cluster, but since \textsc{Avatar} is locally checkable, all faulty configurations of $\textsc{Avatar}(\textsc{Cbt})$ contain at least one node which detects the faulty configuration and will begin forming the single-node clusters.  This fault detection and cluster creation will propagate through the network until eventually all nodes are members of $N$ virtual node \textsc{Cbt} clusters.
    \item \emph{Matching}: The second step of the algorithm is to match together clusters so that they may merge together.  To do this, the root node of the binary tree repeatedly polls the nodes of its cluster, asking them to either find neighboring clusters that are looking for merge partners (called the leader role), or to look at neighboring clusters that can assign them a merge partner (called the follower role).  The role of leader or follower is randomly selected.  Leader clusters will match together all of their followers for merging, creating a matching between clusters that may not be direct neighbors.  This ability to create edges to match non-neighboring clusters allows more matches to occur, and thus more merges, and thus a faster convergence time.
    \item \emph{Merging}: The algorithm then deals with the merging of matched clusters.  To prevent degrees from growing too large, a cluster is only allowed to merge with at most one other cluster at a time.  Once two clusters have matched from the previous step, the roots of the clusters connect as ``partners'' and update their successor pointers based upon the identifier of the host of the root of the other cluster.  One node will have its responsible range become smaller, and this node will send all guest nodes that were in its old responsible range to its partner in the other cluster.  The children of the root nodes are connected, and then they repeat the process of updating successor pointers and passing along guest nodes outside their new responsible range.  Eventually this process reaches the leaves, at which point all nodes in both clusters have updated their responsible ranges and now form a new legal \textsc{Cbt} cluster.
\end{enumerate} 

This process of matching and merging continues until eventually only a single cluster is left, which is the correct $\textsc{Avatar}(\textsc{Cbt})$ network.  We restate the following theorem from the original work and offer a brief sketch of the proof's intuition.

\begin{theorem}
The self-stabilizing algorithm for $\textsc{Avatar}(\textsc{Cbt})$ by Berns~\cite{berns_avatar-2016} has a convergence time of $\mathcal{O}(\log^2 N)$ rounds in expectation, and a degree expansion of $\mathcal{O}(\log^2 N)$ in expectation.
\end{theorem}

\noindent
\emph{Intuition:} A cluster has a constant probability of being matched and merged with another cluster in $\mathcal{O}(\log N)$ rounds, meaning the number of clusters is reduced by a constant fraction every $\mathcal{O}(\log N)$ rounds in expectation.  This matching and merging only needs to happen $\mathcal{O}(\log N)$ times until we have a single cluster, giving us a time complexity of $\mathcal{O}(\log^2 N)$ rounds in expectation.  The degree of a node can grow during a merge or during the matching process.  However, during a merge a node's degree will grow to at most $\mathcal{O}(\log^2 N)$, and the node's degree will increase by only a constant amount during each match and there are only $\mathcal{O}(\log N)$ such matches in expectation, meaning the degree expansion of the algorithm is also $\mathcal{O}(\log^2 N)$.

\subsubsection{Communication}
The original work on \textsc{Avatar}(\textsc{Cbt}) also defined a communication mechanism to execute on the guest \textsc{Cbt} network.  We will also use this mechanism in our algorithm to ensure edges are added systematically and thus limiting unnecessary degree growth.  In particular, we will use a variant of a \emph{propagation of information with feedback} (PIF) algorithm~\cite{pif_jpdc_2010} which will execute on the (guest) nodes of \textsc{Cbt}.  While the original work was snap-stabilizing, this would not be a requirement in our work.  We are instead simply interested in an organized way to communicate information in waves in a tree.

In PIF, communication happens in waves that are initiated by the root of the binary tree.  First, the root executes a \emph{propagate} action which sends information down the tree level by level.  Once the propagated information reaches the leaves, the leaves begin a \emph{feedback} action, performing some operation and then signalling to their parent that the message has been received and acted upon by all descendants in the tree.  Once the parent of a node has received the message, the node may \emph{clean} and prepare for the next wave.  Once the root receives the feedback wave, it knows the message was successfully received and acted upon by all nodes in the tree, and the root may continue with further PIF waves if necessary.

We will use this communication mechanism to add edges to our network to build \textsc{Chord}.  As the PIF process itself is previously defined, we only need to provide the actions each node will perform for each part, as well as any data that is sent.  In particular, we will say that a tree $T$ executes a $PIF(X)$ wave, meaning the root of tree $T$ will signal to its children that a propagation wave has begun with the $PIF(X)$ message.  Furthermore, we will specify the \emph{propagate} action of $a$, which is what each non-root node $a$ should do when it receives the propagation message $PIF(X)$.  We also specify the \emph{feedback} action of $a$, which is the actions each node $a$ should take when it receives acknowledgements from its children that the most-recent propagation wave has completed and the corresponding feedback wave is underway.
\section{\textsc{Avatar}(\textsc{Chord})}
\label{section:algorithm}
In this section we discuss how we can use the existing \textsc{Avatar}(\textsc{Cbt}) self-stabilizing overlay network as the starting point for the efficient creation of a variant of the \textsc{Chord} overlay network.

\subsection{Overview of Our Approach}
Arguably one of the major barriers to the practical implementation for self-stabilizing overlay networks is the complexity that must be managed when designing and analyzing these networks, particularly when we desire efficient self-stabilization.  For instance, TCF~\cite{berns_tcf_11} is simple and works with any locally-checkable topology, but it requires $\Theta(n)$ space.  One could imagine a simple ``design pattern'' which simply suggests that in every round, a node computes their ideal neighborhood given the information available to them from their state and the state of their neighbors, and then add and delete edges to form this ideal neighborhood.  Unfortunately, analyzing this algorithm in terms of \emph{both} correctness and efficiency is quite difficult as one must consider the implications of a variety of actions on a variety of possible initial configurations.

One approach to managing complexity is to start by building smaller or simpler structures, and then using these to continue towards the final goal.  Consider, for instance, the construction of a large building.  One common approach is to erect a simple scaffold and use this scaffold to build the more complex permanent structure.  As another example, consider the prior work on using \emph{convergence stairs} for analyzing general self-stabilizing algorithms.  In this technique, one first must show the system converges to some weaker predicate $A_0$ from an arbitrary initial configuration, then show it converges to $A_1$ provided it is in $A_0$, then show it converges to $A_2$ provided it is in $A_1$, and so on until you have reached the correct configuration.  These patterns of design and analysis are similar in that they take a complex set of required actions, decompose them into smaller distinct steps, and then rely on prior solutions to the smaller steps to move to the next ones.

In the remainder of this section, we discuss our approach for efficiently creating a self-stabilizing version of the \textsc{Chord} network based upon this idea of scaffolding.  In particular, we shall define the \textsc{Chord} topology, and then discuss how we can use $\textsc{Avatar}(\textsc{Cbt})$ as a starting point for constructing \textsc{Chord}.  We then show how nodes can determine in a short amount of time whether they should be building the ``scaffold'' (\textsc{Cbt}) or the target topology (\textsc{Chord}).

\subsection{\textsc{Chord}(N)}
\label{section:target}
Our target network aims to resolve the lack of robustness of the \textsc{Cbt} scaffold network.  In particular, our target network is an $N$-node \textsc{Chord} network defined as follows:

\begin{definition}
For any $N \in \mathbb{N}$, let $\textsc{Chord}(N)$ be a graph with nodes $[N]$ and edge set defined as follows.  For every node $i$, $0 \leq i < N$, add to the edge set $(i, j)$, where $j = (i + 2^k) \mod N$, $0 \leq k < \log N - 1$.  When $j = (i + 2^k)\mod N$, we say that $j$ is the $k$-th finger of $i$.
\end{definition}

It is worth again noting that our use of the \textsc{Avatar} framework results in a locally-checkable version of the \textsc{Chord} network.  \textsc{Chord} as defined on an arbitrary set of nodes is actually \emph{not} locally checkable, particularly because of the ``ring'' edges (in a legal configuration, exactly one node should have two immediate neighbors with smaller identifiers, but which node this should be cannot be determined if the node set is arbitrary).  Unlike prior approaches, then, our stabilizing \textsc{Chord} network is \emph{silent}, meaning no messages or ``probes'' need to be continuously exchanged between nodes in a legal configuration.

Our goal, then, is to use the $N$-node topology of \textsc{Cbt} to add edges to the guest nodes (and to the corresponding host network as required to maintain a dilation-1 embedding) until we have formed the correct $N$-node \textsc{Chord} network.

\subsection{Building \textsc{Chord} from \textsc{Cbt}}
Figure~\ref{algo:chord_scaffold} elaborates on the algorithm which uses our guest \textsc{Cbt} network as a scaffold for creating the guest \textsc{Chord} network.  The algorithm uses the fact that \textsc{Chord} edges can be created inductively.  That is, assuming all fingers from $0$ to $k$ are present, the $k+1$ finger can be created in a single round.  Specifically, if node $b$ is the $(i-1)$ finger of $c_0$, and $c_1$ is the $(i-1)$ finger of $b$, the $i$th finger of $c_0$ is $c_1$.  The algorithm begins by correctly building finger $0$, then recursively adds the first finger, then the second, and so on.  This adding of edges is done in a metered fashion, however, to prevent unnecessary degree growth from faulty initial configurations.

Once the scaffold network has been built, we can begin the process of constructing our final target topology.  We design our algorithm to execute on the $N$ guest nodes of \textsc{Cbt}, with the goal being to add edges to the nodes of \textsc{Cbt} until they have formed the $N$ guest node \textsc{Chord} network.  For now, we shall assume that the network is in the legal \textsc{Cbt} configuration.  We will relax this assumption and consider an arbitrary initial configuration shortly.

The algorithm begins with the root of $\textsc{Cbt}$ initiating a $\mathit{PIF}$ wave which connects each guest node with its $0$th finger.  Notice that, with the exception of one node, the edges in the host network realizing every guest node's $0$th finger are already present.  For any guest node $b \neq N-1$, the $0$th finger of $b$ is either (i) a guest node with the same host as $b$, or (ii) a guest node which is hosted by the successor of $\mathit{host}_b$.  Edges to guest nodes $0$ and $N-1$ are forwarded up the tree during the feedback wave, allowing the root of the tree to connect them at the completion of the wave, thus forming the base ring and completing every guest node's $0$th finger.  The root then executes $\log N-1$ additional $\mathit{PIF}$ waves, with wave $k$ correctly adding the $k$th finger for all guest nodes.  After $\mathcal{O}(\log^2 N)$ rounds, we have built the correct $\textsc{Avatar}(\textsc{Chord})$ network.

\begin{figure}
\begin{tabbing}
........\=....\=....\=....\=....\=....\=....\kill
\>\textit{// Execute when $\mathit{phase}_u = \mathit{CHORD}$; If $\mathit{phase}_u = \mathit{CBT}$, then execute} \\
\>\textit{// the original $\textsc{Avatar}(\textsc{Cbt})$ algorithm~\cite{berns_avatar-2016}}.\\
\>\textit{// As part of each round, nodes exchange their local state, including}\\
\>\textit{// $LastWave$, and check for faulty configurations as described in Section \ref{sect:phase}.}\\
1.\>Tree $T$ executes a $\mathit{PIF}(\mathit{MakeFinger}(0))$ wave:\\
2.\>\>\textbf{Propagate Action for $a$:} $\mathit{LastWave}_a = 0$\\
3.\>\>\textbf{Feedback Action for $a$:}\\
\>\>\>\textit{// Let $b$ be the $0$th finger of $a$.}\\
4.\>\>\>\textbf{if} $\mathit{LastWave}_a = \mathit{LastWave}_b = 0$ \textbf{then}\\
5.\>\>\>\>Create the edge $(a,b)$\\
6.\>\>\>\>Forward an edge to node $0$ \\
  \>\>\>\>\>\>or $N-1$ (if present) to parent\\
7.\>\>\>\textbf{else} $\mathit{phase}_u = \mathit{CBT}$ (where $u$ is $\mathit{host}_a$) \textbf{fi}\\
8.\>\textbf{for} $k = 1,2,\ldots,\log N - 1$ \textbf{do}\\
9.\>\>Tree $T$ executes a $\mathit{PIF}(\mathit{MakeFinger}(k))$ wave:\\
10.\>\>\>\textbf{Propagate Action for $a$:} $\mathit{LastWave}_a = k$\\
11.\>\>\>\textbf{Feedback Action for $a$:}\\
\>\>\>\>\textit{// Let $b_0, b_1$ be the $k-1$ fingers of $a$.}\\
12.\>\>\>\>\textbf{if} $\mathit{LastWave}_a = \mathit{LastWave}_{b_0} = \mathit{LastWave}_{b_1} = k$ \textbf{then}\\
13.\>\>\>\>\>Create edge $(b_0,b_1)$, the $k$th finger of $b_0$.\\
14.\>\>\>\>\textbf{else} $\mathit{phase}_u = \mathit{CBT}$ (where $u$ is $\mathit{host}_a$) \textbf{fi}\\
15.\>\textbf{od}
\end{tabbing}
\caption{Algorithm 1: PIF for \textsc{Chord} Target from \textsc{Cbt} Scaffold}
\label{algo:chord_scaffold}
\end{figure}

\subsection{Phase Selection}
\label{sect:phase}
The final piece for our self-stabilizing \textsc{Chord} network is to create a mechanism by which nodes can know which algorithm they should be executing: either executing the steps required to build the $\textsc{Avatar}(\textsc{Cbt})$ network, or the steps required to build the \textsc{Chord} target network from an existing \textsc{Cbt} network).  We assume each host node $u$ maintains a phase variable $\mathit{phase}_u$ whose value is from the set $\{CBT, CHORD, DONE\}$.  When $\mathit{phase}_u = \mathit{CBT}$, a node is executing the algorithm for the $\textsc{Avatar}(\textsc{Cbt})$ network.  If $\mathit{phase}_u = \mathit{CHORD}$, then the PIF waves in Algorithm \ref{algo:chord_scaffold} are executed.  If $\mathit{phase}_u = \mathit{DONE}$, then a node will take no actions provided its local neighborhood is consistent with a legal $\textsc{Avatar}(\textsc{Chord})$ network.

Determining which algorithm to execute requires a node be able to determine if the configuration they are in now is either completely correct or consistent with one reached by building \textsc{Chord} from \textsc{Cbt}.  We define a subset of states under which Algorithm \ref{algo:chord_scaffold} will converge, and then define a predicate which nodes can use to determine if the network is in one of these states.

\begin{definition}
A graph $G$ with node set $V$ is in a \emph{scaffolded \textsc{Chord} configuration} if $G$ is reachable by executing the PIF waves defined by Algorithm \ref{algo:chord_scaffold} on a correct $\textsc{Avatar}(\textsc{Cbt})$ network.
\end{definition}

Thanks to the predictability of the \textsc{Cbt} scaffold network, nodes can determine if their state is consistent with that of a scaffolded \textsc{Chord} configuration.  Informally, each guest node can determine this by simply checking to see if its neighborhood is a superset of $\textsc{Cbt}$ but a subset of $\textsc{Chord}$, with the first $k$ fingers from $\textsc{Chord}$ present, for some $k \in [0, \log N)$.  We define the predicate a node can use for this operation below.

\begin{definition}
Let $\mathit{scaffolded}_b$ be a predicate defined over the local state of a guest node $b$, as well as the state of nodes $b' \in N(b)$.  The value of $\mathit{scaffolded}_b$ is the conjunction of the following conditions.
\begin{enumerate}
\item Node $b$ has all neighbors from $\textsc{Cbt}$, each with the proper host and tree identifier (a value set as part of a legal \textsc{Cbt} scaffold network).
\item Node $b$ has last executed the $k$th feedback wave of a\\ $\mathit{PIF}(\mathit{MakeFinger}(k),\bot)$ wave, for some $0 \leq k < \log N$
\item All neighbors of $b$ have either all $k$ fingers present, or $k+1$ fingers (if a child has just processed a feedback wave), or $k-1$ (if parent has not yet processed the current feedback wave), where $k$ is the last feedback wave $b$ has executed
\item Node $b$'s parent has last executed the $k$th feedback wave, and has the first $k$ $\textsc{Chord}$ fingers, or $k-1$ fingers if $b$ has just completed the feedback transition and $b$'s parent has not.
\end{enumerate}
\end{definition}

In every round of computation, all nodes are checking their local state and the state of their neighbors to determine if a faulty configuration is found.  This check for faults, along with the $\mathit{scaffolded}_b$ predicate, is used to set the $\mathit{phase}_u$ variable as follows.  If a fault is detected and $\mathit{scaffolded}_b = \mathit{false}$, then $u = \mathit{host}_b$ sets $\mathit{phase}_u = \mathit{CBT}$.  Furthermore, if any neighbor $v$ has a different value for $\mathit{phase}_v$, then $\mathit{phase}_u = \mathit{CBT}$.  We will show in a moment that this procedure is sufficient to ensure the correct algorithm is executed within a short amount of time (i.e. if the configuration is not a scaffolded \textsc{Chord} configuration, then all nodes begin executing the $\textsc{Avatar}(\textsc{Cbt})$ algorithm quickly).  Notice that once the correct configuration is built, nodes can execute a final $\mathit{PIF}$ wave to set $\mathit{phase}_u = \mathit{DONE}$.  If any node detects \emph{any} fault during this process, it simply sets $\mathit{phase}_u = \mathit{CBT}$.  Since $\textsc{Avatar}(\textsc{Chord})$ is locally checkable, at least one node will not set $\mathit{phase}_u = \mathit{DONE}$ during the final $\mathit{PIF}$ wave, and the $\textsc{Avatar}(\textsc{Cbt})$ algorithm will begin.

\section{Analysis}
\label{section:analysis}
We sketch the proofs for our main results below.  The full proofs of convergence and degree expansion are contained in the appendix.

\begin{theorem}
Algorithm \ref{algo:chord_scaffold}, when combined with the self-stabilizing algorithm for $\textsc{Avatar}(\textsc{Cbt})$ from Berns~\cite{berns_avatar-2016}, is a self-stabilizing algorithm for the network $\textsc{Avatar}(\textsc{Chord})$ with convergence time $\mathcal{O}(\log^2 N)$ in expectation.
\end{theorem}
\begin{proofsketch} 
To prove the convergence time of our algorithm, we first show that if the configuration is not a scaffolded \textsc{Chord} configuration, within $\mathcal{O}(\log N)$ rounds, all nodes are executing the algorithm to build the scaffold \textsc{Cbt} network.  We then show that nodes will have built the correct \textsc{Cbt} network within an additional $\mathcal{O}(\log^2 N)$ rounds in expectation, at which point all nodes begin building the target \textsc{Chord} network.  We will then show that this process succeeds in $\mathcal{O}(\log^2 N)$ rounds.  Putting these together, we get an overall convergence time of $\mathcal{O}(\log^2 N)$ in expectation.
\end{proofsketch}
\begin{theorem}
Algorithm \ref{algo:chord_scaffold}, when combined with the self-stabilizing algorithm for $\textsc{Avatar}(\textsc{Cbt})$ from Berns~\cite{berns_avatar-2016}, is a self-stabilizing algorithm for the network $\textsc{Avatar}(\textsc{Chord})$ with degree expansion of $\mathcal{O}(\log^2 N)$ in expectation.
\end{theorem}
\begin{proofsketch} 
By design, any edge that is added to the network when building \textsc{Chord} from \textsc{Cbt} is an edge that will remain in the final correct configuration and therefore does not affect the degree expansion.  Furthermore, we know from the original \textsc{Avatar} paper that the expected degree expansion is $\mathcal{O}(\log^2 N)$ when all nodes are executing the \textsc{Cbt} algorithm.

The only new piece we need to consider, then, is to analyze the actions nodes might take when they incorrectly believe, based on their local state, that they are building the \textsc{Chord} network from the \textsc{Cbt} scaffold (a ``false \textsc{Chord}'' phase), which we show can only happen for $\mathcal{O}(\log N)$ rounds.  Since adding \textsc{Chord} edges is coordinated with a PIF wave, each guest node $b$ can only increase its degree by one during this time.  At most, then, a node may increase its degree by a factor of 2 during this time, leading to the initial degree growth of $2$ during the ``false \textsc{Chord}'' phase.
\end{proofsketch}
\section{Generalizing Our Approach}
\label{section:overview}
Above we have provided an algorithm for using one self-stabilizing overlay network to create another self-stabilizing overlay network.  While we are not the first to use this general idea in the construction of overlay networks, we are the first to explicitly define and discuss this approach, which we call \emph{network scaffolding}.  To use the network scaffolding approach, one must define several components.  In particular, we must define:
\begin{itemize}
    \item The \emph{scaffold network}, an intermediate topology which we can construct from any initial configuration.
    \item The \emph{target network}, the network topology that we wish to build for use with our final application.
    \item A self-stabilizing algorithm for constructing the scaffold network.
    \item An algorithm for building the target network when starting from the correct scaffold network.
    \item A local predicate allowing nodes to determine whether they should be building the scaffold network or the target network.
\end{itemize}

Our self-stabilizing algorithm from above used $\textsc{Avatar}(\textsc{Cbt})$ as the scaffold network to build a $\textsc{Avatar}(\textsc{Chord})$ target network.  To do was relatively straightforward: we defined a way to build \textsc{Chord} from \textsc{Cbt}, and then proved nodes would quickly determine which network they were building.

This network scaffolding approach has been used in some form by other previous work, and we hope it will be extended in future work as well.  Our approach heavily depends upon the scaffold network selected.  The \textsc{Cbt} network has many desirable properties for a scaffold network when compared to other examples in prior work.  These properties include:

    \noindent
    \textbf{Efficient self-stabilization:}  If the scaffold itself is inefficient to build, we cannot expect the target topology to be built efficiently.  TCF~\cite{berns_tcf_11} can be thought of as an inefficient scaffold network that requires $\mathcal{O}(n)$ space.  \textsc{Avatar}(\textsc{Cbt}) is a logical choice as, prior to this work, it is the only self-stabilizing overlay network we are aware of with efficient stabilization in terms of both time and space.
    
    \noindent
    \textbf{Low node degree:}  Unlike a real scaffold, we maintain the scaffold edges after the target network is built.  Therefore, the scaffold network must have low degree if we wish our final configuration to be so.  Again, the suitability of \textsc{Avatar}(\textsc{Cbt}) is apparent, as it requires only a few edges per virtual node (and a logarithmic number of edges per real node).
    
    \noindent
    \textbf{Low diameter:}  Low diameter allows (relatively) fast communication for adding the target network's edges one at a time.  A previous work, \textsc{Re-Chord}~\cite{kniesburges_rechord_11}, used a ``scaffold'' of the \textsc{Linear} network, whose $\mathcal{O}(n)$ diameter contributed to the $\mathcal{O}(n \log n)$ convergence time of their algorithm.
    
    \noindent
    \textbf{Predictable routing:}  The predictable routing, particularly for communication, allows us to add edges in a metered and checkable fashion.  This predictability helps with both design and analysis.  It would be interesting to see if a semi-structured overlay network could be used as a scaffold, as semi-structured overlays may be easier to build.  To date, little work has been done on self-stabilizing semi-structured overlay networks, but there are several examples of efficient creation of semi-structured networks in non-self-stabilizing settings~\cite{gotte_time-optimal_podc21} which may be interesting starting points for future work.
    
    \noindent
    \textbf{Local checkability:} To be able to determine which phase of the algorithm should be executed quickly, without ``wasting'' time and resources adding edges from a faulty configuration, the scaffold should ideally be locally checkable.  Some previous overlay networks have used a ``probing'' approach where messages were circulated continuously to try and detect faulty configurations.  The risk of this approach in network scaffolding is that nodes may spend too long adding edges from an incorrect scaffold, or take too long to detect a faulty configuration.
%
%
%
%
%
%
%
%
%
%
%
%
%
%
%
%
%
%
%
%
%
%
%
%
%
%
%
%
%
%
%
%
%
%
%
%
%
%
%
%
%
%
%
%
%
%
%
%
%
%
%
%
%
%
%
\section{Concluding Thoughts}
\label{section:concluding_thoughts}
In this paper, we have presented the first time- and space-efficient algorithm for building a \textsc{Chord} network using a technique we call \emph{network scaffolding}.  We discussed considerations for expanding this technique, in particular pointing out considerations and implications for various properties of the scaffold network.

An obvious extension to our work would be to consider building other target topologies using \textsc{Avatar}(\textsc{Cbt}) as a scaffold network.  For instance, networks with good load balancing properties or with high resilience to churn could be converted into self-stabilizing variants using \textsc{Avatar} to define the network topology and the \textsc{Cbt} scaffold to build this correct topology.  It would also be interesting to investigate the correctness and complexity of this approach when using a more realistic asynchronous communication model.


\bibliographystyle{splncs04}
\bibliography{overlays}
\appendix
\appendixpage
\let\cleardoublepage\clearpage
\section{Convergence Time Analysis}
First, consider how much time is required for the network to begin the algorithm for building the correct \textsc{Cbt} network.  We will show that, after $\mathcal{O}(\log N)$ rounds, if the configuration is neither the correct $\textsc{Avatar}(\textsc{Chord})$ network nor a scaffolded $\textsc{Chord}$ configuration, then all nodes are executing the $\textsc{Avatar}(\textsc{Cbt})$ algorithm.  During these $\mathcal{O}(\log N)$ rounds, the degree of a node can at most double from its initial degree.  Finally, after the correct $\textsc{Avatar}(\textsc{Cbt})$ network is built, the $\textsc{Avatar}(\textsc{Chord})$ network is built in an additional $\mathcal{O}(\log^2 N)$ rounds.  We prove each of these claims below.

\begin{lemma}
Let $G_i$ be a configuration with node set $V \subseteq [N]$.  Suppose $G_i$ is not a legal $\textsc{Avatar}(\textsc{Chord})$ configuration, and $G_i$ is also not a scaffolded $\textsc{Chord}$ configuration.  After 1 round, any node $u \in G_i$ is distance at most $2 \cdot (\log N + 1)$ from a node $v$ with $\mathit{phase}_v = \mathit{CBT}$.
\label{lemma:detdiam_intermediate_chord}
\end{lemma}
\begin{proof}
Imagine the subgraph induced by all nodes at distance at most $2 \cdot (\log N + 1)$ from any node $u$.  If this ball does not contain all the nodes in the network, then there exists at least one node $v$ in this ball such that $v$ has a neighbor that it would not have in the correct $\textsc{Avatar}(\textsc{Cbt})$ configuration, and $v$ detects this extra neighbor.  Therefore, $v$ would set $\mathit{phase}_{v} = CBT$ in a single round.  Furthermore, the host nodes in this ball must be hosting exactly $N$ guest nodes, with each guest node hosted by the correct host node with the correct tree identifier, else at least one node detects an incorrect embedding (by prior result on the local checkability of \textsc{Avatar}).  Finally, the $\mathit{PIF}$ state of these $N$ guest nodes must be consistent, else a faulty configuration that is not a scaffolded \textsc{Chord} configuration is detected, and $\mathit{phase}_u = \mathit{CBT}$.

Assume the $N$ guest nodes in this ball realize a $\textsc{Cbt}$ network with some additional edges, and the $\mathit{PIF}$ state is consistent.  If the network is a correct $\textsc{Avatar}(\textsc{Chord})$ configuration, no node detects a fault, as $\textsc{Avatar}(\textsc{Chord})$ is locally checkable.  A guest node can easily determine if it has the first $0 \leq k < \log N$ $\textsc{Chord}$ fingers by examining its set of neighbors.  Guest nodes can also verify that the last $\mathit{PIF}$ feedback wave processed was for this $k$ (potentially $k-1$ if finger $k$ is a member of $\textsc{Cbt}$).  If a $b$ node detects it has a different number of correct \textsc{Chord} fingers as a neighbor, it must either be currently processing a $\mathit{PIF}$ wave and about to receive the feedback wave (if the children have an extra finger) or pass the feedback wave to its parent (if its parent has one less correct \textsc{Chord} finger).  If neither of these conditions are true, then $b$ sets $\mathit{scaffolded}_b = 0$.  Therefore, if all nodes do not have their first $k$ \textsc{Chord} fingers correct, or all nodes are not in a $\mathit{PIF}$ wave to add their $k$th finger to the already-correct set of $k-1$ fingers, at least one node must set $\mathit{scaffolded}_b = 0$, which results in $\mathit{phase}_u = \mathit{CBT}$ (where $u = \mathit{host}_b$).
\end{proof}

\begin{lemma}
Suppose configuration $G_i$ is not the $\textsc{Avatar}(\textsc{Chord})$ configuration, and $G_i$ is not a scaffolded \textsc{Chord} configuration.  In at most $2 \cdot (\log N + 1)$ rounds, all nodes are executing the self-stabilizing $\textsc{Avatar}(\textsc{Cbt})$ algorithm.
\label{lemma:time_to_cbt}
\end{lemma}
\begin{proof}
By Lemma \ref{lemma:detdiam_intermediate_chord}, if $G$ is not $\textsc{Avatar}(\textsc{Chord})$ nor a scaffolded $\textsc{Chord}$ network, then every node $u$ has a node $v$ within distance at most $2 \cdot (\log N + 1)$ with $\mathit{phase}_v = \mathit{CBT}$.  In every round, neighbors of nodes with $\mathit{phase}_v = \mathit{CBT}$ set their own $\mathit{phase}_u = \mathit{CBT}$.  In at most $2 \cdot (\log N + 1)$ rounds, every node $u$ has $\mathit{phase}_u = \mathit{CBT}$, thus every node is executing the $\textsc{Avatar}(\textsc{Cbt})$ algorithm.
\end{proof}

At this point, we know that either the network is the correct target network or all nodes are executing the algorithm for building the \textsc{Cbt} guest network.  We restate the convergence time for this algorithm below, as the full proof is contained in the original \textsc{Avatar} work~\cite{berns_avatar-2016}.

\begin{theorem}
\label{thm:cbt_convergence}
From any configuration where all nodes are executing the self-stabilizing algorithm for \textsc{Cbt} (that is, all nodes have $\mathit{phase}_u = \mathit{CBT}$), in $\mathcal{O}(\log^2 N)$ rounds in expectation, the correct $\textsc{Avatar}(\textsc{Cbt}(N))$ is built.
\end{theorem}

By Lemma \ref{lemma:time_to_cbt} and Theorem \ref{thm:cbt_convergence}, after an expected $\mathcal{O}(\log^2 N)$ rounds, the scaffold network has been built.  Our last step is simply to analyze the time required for the \textsc{Chord} target network to be built using the \textsc{Cbt} scaffold, which we do next.

\begin{lemma}
Let $G_0$ be the correct $\textsc{Avatar}(\textsc{Cbt})$ configuration, and let each node $u \in V$ have $\mathit{phase}_u = \mathit{CHORD}$.  In $\mathcal{O}(\log^2 N)$ rounds, the configuration converges to $\textsc{Avatar}(\textsc{Chord})$.
\end{lemma}
\begin{proof}
We prove this by induction on the number of correct fingers.  In $G_0$, the root of the guest network $\textsc{Cbt}$ will begin the first $\mathit{PIF}$ wave to add the $0$th fingers.  For every guest node $b \neq N-1$, the $0$th finger is either (i) a guest node hosted by $u = \mathit{host}_b$, or (ii) a guest node hosted by $v$ such that $v = \mathit{succ}_u$.  In either case, the edge in the host network realizing this guest edge already exists, and creating this $0$th finger is simply a matter of a host updating a guest node's local state in response to the $\mathit{PIF}$ wave.  For node $N-1$, the first $\mathit{PIF}$ wave forwards to the root node an edge to guest nodes $0$ and $N-1$, and the root will connect these two nodes.  This requires $2 \cdot (\log N + 1)$ rounds.  At this point, all nodes have the correct $0$th finger.

Assume finger $i$ has been created by the $\mathit{PIF}(\mathit{Chord}(i),\bot)$ wave.  The root will then execute the $\mathit{PIF}(\mathit{Chord}(i+1),\bot)$ wave.  Let guest node $b$ receive the feedback wave.  By the inductive hypothesis, node $b$ has links to guest nodes $c_0$ and $c_1$, where $b$ is finger $i$ of $c_0$, and $c_1$ is finger $i$ of $b$.  Node $b$ connects $c_0$ and $c_1$ in one round, thereby creating finger ($i+1$) for $c_0$.  After $2 \cdot (\log N + 1)$ rounds, all nodes have added their $i+1$ finger.  As there are $\log N - 1$ fingers, the total convergence time to reach the $\textsc{Avatar}(\textsc{Chord})$ network from $G_0$ is $\mathcal{O}(\log^2 N)$.
\end{proof}

Using Theorem \ref{thm:cbt_convergence} and the above lemmas gives us the following result.

\begin{theorem}
Algorithm \ref{algo:chord_scaffold} combined with the self-stabilizing algorithm for $\textsc{Avatar}(\textsc{Cbt})$ from the original work~\cite{berns_avatar-2016} is a self-stabilizing algorithm for the $\textsc{Avatar}(\textsc{Chord})$ network with convergence time $\mathcal{O}(\log^2 N)$ in expectation.
\end{theorem}
\section{Degree Expansion Analysis}
Besides converging to a correct configuration quickly, the algorithm maintains low degree expansion.  By design, any edge that is added to the network when building \textsc{Chord} from \textsc{Cbt} is an edge that will remain in the final correct configuration and therefore does not affect the degree expansion.  Furthermore, we know from the original \textsc{Avatar} paper that the degree expansion is polylogarithmic when all nodes are executing the \textsc{Cbt} algorithm.  We restate this theorem from the original work below.

\begin{theorem}
\label{thm:cbt_degree}
When starting from an arbitrary configuration where all nodes are executing the \textsc{Cbt} algorithm (i.e. $\mathit{phase}_u = \mathit{CBT}$ for all $u \in V$), the degree expansion is $\mathcal{O}(\log^2 N)$ in expectation.
\end{theorem}

The only new piece we need to provide here is to analyze the actions nodes might take when they incorrectly believe based on their local state that they are building the \textsc{Chord} network from the \textsc{Cbt} scaffold.  Earlier we had mentioned that there is a benefit from executing our algorithm on the larger $N$ node guest network, even though the convergence time might be slower.  It is here that we see this benefit: because adding an edge requires we coordinate with all $N$ guest nodes over $\mathcal{O}(\log N)$ rounds, if the network is not in a scaffolded \textsc{Chord} configuration, all nodes will detect this before too many unnecessary edges are added to the network.  We sketch the proof for this below.

\begin{lemma}
Suppose $G_i$ is neither a correct $\textsc{Avatar}(\textsc{Chord})$ configuration nor a scaffolded \textsc{Chord} configuration.  The degree of any node $u \in V$ increases by at most a factor of $2$ before $u$ begins executing the $\textsc{Avatar}(\textsc{Cbt})$ algorithm.
\label{lemma:degree_expansion}
\end{lemma}
\begin{proof}
Consider a guest node $b$ hosted by $u$ in configuration $G_i$.  Notice that each guest neighbor $b'$ of $b$ can add at most one node to the neighborhood of $b$ per $2 \cdot (\log N + 1)$ rounds when executing Algorithm \ref{algo:chord_scaffold}, as $b'$ attempts to add the $k$th finger to $b$ (at which point $b'$ must wait for another $\mathit{PIF}$ wave).  Furthermore, by Lemma \ref{lemma:time_to_cbt}, after at most $2 \cdot (\log N + 1)$ rounds, all nodes in $G_i$ have $\mathit{phase}_u = \mathit{CBT}$.  Therefore, in the worst case a node $u$ may receive a single edge from all neighbors from configuration $G_i$, resulting in a degree expansion of at most $2$.
\end{proof}

Combining the lemma above with Theorem \ref{thm:cbt_degree} gives us the following theorem.

\begin{theorem}
Algorithm \ref{algo:chord_scaffold} builds the $\textsc{Avatar}(\textsc{Chord})$ network with degree expansion of $\mathcal{O}(\log^2 N)$ in expectation.
\end{theorem}
\begin{proof}
By Lemma \ref{lemma:degree_expansion}, the degree expansion from Algorithm \ref{algo:chord_scaffold} is at most $2$ before the $\textsc{Avatar}(\textsc{Cbt})$ algorithm is executed.  By Theorem \ref{thm:cbt_degree}, this algorithm has a degree expansion of $\mathcal{O}(\log^2 N)$ in expectation.  When executing from a scaffolded \textsc{Chord} configuration, the only edge added by Algorithm \ref{algo:chord_scaffold} that is not an edge in the final configuration is the forwarding of an edge to node $0$ and $N-1$.  Therefore, once the correct $\textsc{Avatar}(\textsc{Cbt})$ network is built, Algorithm \ref{algo:chord_scaffold} builds $\textsc{Avatar}(\textsc{Chord})$ with degree expansion of at most 2.  Therefore, from any configuration Algorithm \ref{algo:chord_scaffold} may increase degrees by a factor of at most $2$, the original \textsc{Cbt} algorithm may increase degrees by $c \cdot \log^2 N$, giving us $\mathcal{O}(\log^2 N)$ degree expansion.
\end{proof}
\end{document}